\documentclass[conference]{IEEEtran} \IEEEoverridecommandlockouts

\usepackage{amsmath, amsfonts, dsfont}     
\usepackage{amssymb}               
\usepackage{cases}                 
\usepackage{relsize}               
\usepackage{mathbbol}

\allowdisplaybreaks[4]

\makeatletter
\newcommand{\biggg}{\bBigg@{3}}

\newcommand{\Biggg}{\bBigg@{3.5}}

\newcommand{\bigggg}{\bBigg@{4}}

\newcommand{\Bigggg}{\bBigg@{4.5}}

\makeatother

\usepackage{algorithm}            
\usepackage{algorithmic}          

\usepackage{graphicx}             
\graphicspath{{./figs/}}          

\usepackage{adjustbox}            

\usepackage{makecell}            
\usepackage{array}               
\usepackage{threeparttable}      
\usepackage{stfloats}            

\usepackage[caption=false,font=normalsize]{subfig}

\usepackage{xcolor}              
\usepackage{color}               

\usepackage{textcomp}            
\usepackage{url}                 
\usepackage{verbatim}            
\usepackage{cite}                
\usepackage{enumitem}            

\usepackage{lineno}              

\let\oldequation\equation
\let\oldendequation\endequation
\renewenvironment{equation}
   {\linenomathNonumbers\oldequation}
   {\oldendequation\endlinenomath}
\let\oldalign\align
\let\oldendalign\endalign
\renewenvironment{align}
   {\linenomathNonumbers\oldalign}
   {\oldendalign\endlinenomath}
\let\oldnumcases\numcases
\let\oldendnumcases\endnumcases
\renewenvironment{numcases}[1]
   {\begingroup\linenomathNonumbers\oldnumcases{#1}}
   {\oldendnumcases\endlinenomath\endgroup}

\newcommand{\tr}{\mathsf{tr}}
\newcommand{\rank}{\mathsf{r}}

\renewcommand{\det}{\mathsf{det}}

\def\mindex#1{\index{#1}}

\def\sq{\hbox{\rlap{$\sqcap$}$\sqcup$}}
\def\qed{\ifmmode\sq\else{\unskip\nobreak\hfil
\penalty50\hskip1em\null\nobreak\hfil\sq
\parfillskip=0pt\finalhyphendemerits=0\endgraf}\fi\medskip}

\long\def\defbox#1{\framebox[.9\hsize][c]{\parbox{.85\hsize}{%
\parindent=0pt
\baselineskip=12pt plus .1pt      
\parskip=6pt plus 1.5pt minus 1pt 
#1}}}

\long\def\beginbox#1\endbox{\subsection*{}%
\hbox{\hspace{.05\hsize}\defbox{\medskip#1\bigskip}}%
\subsection*{}}
\def\endbox{}

\newsavebox{\junk}
\savebox{\junk}[1.6mm]{\hbox{$|\!|\!|$}}

\def\bfmath#1{{\mathchoice{\mbox{\boldmath$#1$}}%
{\mbox{\boldmath$#1$}}%
{\mbox{\boldmath$\scriptstyle#1$}}%
{\mbox{\boldmath$\scriptscriptstyle#1$}}}}

\def\bfmY{\bfmath{Y}}

\def\bfmhhaY{\bfmath{\hhaY}} 
\def\bfmhhaY{\hbox to 0pt{$\widehat{\bfmY}$\hss}\widehat{\phantom{\raise 1.25pt\hbox{$\bfmY$}}}}

\def\til={{\widetilde =}}

\def\FRAC#1#2#3{\genfrac{}{}{}{#1}{#2}{#3}}

\def\ddtp{{\mathchoice{\FRAC{1}{d^{\hbox to 2pt{\rm\tiny +\hss}}}{dt}}%
{\FRAC{1}{d^{\hbox to 2pt{\rm\tiny +\hss}}}{dt}}%
{\FRAC{3}{d^{\hbox to 2pt{\rm\tiny +\hss}}}{dt}}%
{\FRAC{3}{d^{\hbox to 2pt{\rm\tiny +\hss}}}{dt}}}}

\def\average#1,#2,{{1\over #2} \sum_{#1}^{#2}}

\def\eye(#1){{\bf(#1)}\quad}

\usepackage{amsthm}

\newtheoremstyle{mybold}
      {3pt}
      {3pt}
      {\itshape}
      {}
      {\bfseries}
      {.}
      { }
      {}

\newtheoremstyle{myremark}
      {3pt}
      {3pt}
      {\normalfont}
      {}
      {\bfseries}
      {.}
      { }
      {}

\theoremstyle{mybold}
\newtheorem{theorem}{Theorem}

\newtheorem{proposition}{Proposition}

\theoremstyle{myremark}
\newtheorem{remark}{Remark}

\def\eq#1/{(\ref{eq:#1})}

\newcommand{\beqn}[1]{\notes{#1}%
\begin{eqnarray} \elabel{#1}}

\newcommand{\eeqn}{\end{eqnarray} }

\newcommand{\beq}[1]{\notes{#1}%
\begin{equation}\elabel{#1}}

\newcommand{\eeq}{\end{equation}}

\def\bdes{\begin{description}}
\def\edes{\end{description}}

\newcounter{rmnum}

\newcounter{anum}

{\end{list}}

\def\ass(#1:#2){(#1\ref{#1:#2})}

\def\ritem#1{
\item[{\sf \ass(\current_model:#1)}]
}

\newenvironment{recall-ass}[1]{%
\begin{description}
\def\current_model{#1}}{
\end{description}
}

\newcounter{problem}
\newcounter{save@equation}
\newcounter{save@problem}


\usepackage[most]{tcolorbox}  
\newcounter{ideacounter} 
\renewcommand{\theideacounter}{\arabic{ideacounter}} 

\newcounter{draftcounter} 
\renewcommand{\thedraftcounter}{\arabic{draftcounter}} 

\begin{document}

\title{FDD CSI Feedback under Finite Downlink Training: A Rate-Distortion Perspective}

\author{Shuao Chen, Junyuan Gao, Yuxuan Shi, Yongpeng Wu, Giuseppe Caire, H. Vincent Poor, and Wenjun Zhang 
\thanks{S. Chen, Y. Wu, and W. Zhang are with the Department of Electronic Engineering, Shanghai Jiao Tong
   University, Shanghai 200240, China (e-mails:
   \{shuao.chen, yongpeng.wu, zhangwenjun\}@sjtu.edu.cn)
   (Corresponding author: Yongpeng Wu).
}
\thanks{
   J. Gao is with the Department of Electrical and Electronic Engineering, The Hong Kong Polytechnic University, Hong Kong SAR, China (e-mail: junyuan.gao@polyu.edu.hk).
}
\thanks{
   Y. Shi is with the Department of Networked Intelligence, Pengcheng Laboratory, Shenzhen 410083, China (e-mail: shiyx01@pcl.ac.cn). 
}
\thanks{G. Caire is with the Communications and Information Theory Group, Technische Universit{\"a}t Berlin, Berlin 10587, Germany (e-mail: caire@tu-berlin.de).
}
\thanks{H. V. Poor is with the Department of Electrical and Computer Engineering, Princeton University, Princeton, NJ 08544, USA (e-mail: poor@princeton.edu).
}
}

\maketitle


\begin{abstract}
This paper establishes the theoretical limits of channel state information (CSI) feedback in frequency-division duplexing (FDD) multi-antenna orthogonal frequency-division multiplexing (OFDM) systems under finite-length training with Gaussian pilots. The user employs minimum mean-squared error (MMSE) channel estimation followed by asymptotically optimal uplink feedback. Specifically, we derive a general rate-distortion function (RDF) of the overall CSI feedback system. We then provide both non-asymptotic bounds and asymptotic scaling for the RDF under arbitrary downlink signal-to-noise ratio (SNR) when the number of training symbols exceeds the antenna dimension.
A key observation is that, with sufficient training, the overall RDF converges to the direct RDF corresponding to the case where the user has full access to the downlink CSI. More importantly, we demonstrate that even at a fixed downlink SNR, the convergence rate is inversely proportional to the training length.
The simulation results show that our bounds are tight, and under very limited training, the deviation between the overall RDF and the direct RDF is substantial.
\end{abstract}


\section{Introduction}

The sixth-generation (6G) wireless networks are envisioned to provide ultra-reliable and high-capacity communications by leveraging large-scale antenna arrays and accurate channel state information (CSI). CSI plays a critical role in beamforming, interference management, and resource allocation~\cite{guo2022overview}. In frequency-division duplexing (FDD) systems, the downlink CSI is not directly available at the base station (BS). Instead, the user equipment (UE) estimates the downlink channel and feeds it back to the BS. As the number of BS antennas increases, the required downlink training overhead increases significantly, which becomes a major bottleneck for FDD massive MIMO systems. 
Existing literature~\cite{wen2018deep,guo2022overview} have explored deep learning techniques to reduce the training and feedback cost. While these methods demonstrate promising performance in simulations, the fundamental theoretical limits of CSI feedback have not been fully established.

The CSI feedback problem can essentially be regarded as a lossy compression problem. \textit{Rate-distortion theory} provides a well-established framework to quantify the tradeoff between compression rate and reconstruction distortion~\cite{thomas2006elements}.
Some recent works have applied rate-distortion theory to obtain the fundamental limits of CSI feedback systems~\cite{khalilsarai2023fdd,kim2025fundamental,song2025downlink}.
In particular, the asymptotic behavior of the overall feedback distortion versus the downlink signal-to-noise ratio (SNR) under finite training was established in~\cite{khalilsarai2023fdd}. 
For practical implementations, several effective schemes were proposed. \cite{khalilsarai2023fdd} showed that entropy-coded scalar quantization (ECSQ) performed well, while \cite{song2025downlink} exploited a proper format of the eigenspace of the received training signals and proposed the truncated Karhunen-Loève (TKL) feedback scheme (see~\cite[Sec.~III-C]{song2025downlink} for details).
However, at the theoretical level, the rate-distortion tradeoff under arbitrary downlink SNR and with finite-length pilots remains an open problem.

In this paper, we derive a general and unified overall rate-distortion function (RDF) for an FDD CSI feedback system with finite-length downlink training, minimum mean-squared error (MMSE) estimation at the UE, and asymptotically optimal uplink feedback.
Specifically, we show that the overall closed-form RDF of the CSI feedback system can be explicitly expressed in terms of the eigenspace representation of the received signals.
Under Gaussian pilots with arbitrary downlink SNR, we provide tight and non-asymptotic bounds on the overall RDF, when the number of training symbols exceeds the number of antennas and pilots are placed on all subcarriers.
Furthermore, we show that as the training length increases, the overall RDF asymptotically converges to the direct RDF, corresponding to the ideal case where the UE has full access to the true downlink CSI for feedback. More importantly, we prove that the gap between the two decays inversely with training length.
It should be noted that the RDF is a theoretical bound achievable only with long-delay joint encoding of many pilot observations. In practice, feedback must be provided immediately after each coherence block, and the optimal one-shot rate-distortion tradeoff remains unknown.

In addition to the overall RDF, we focus on the performance of the MMSE estimator under limited observations at the UE. It is well known that, under the same distortion constraint, compressing the MMSE estimate reduces the RDF compared to compressing the original CSI~\cite{berger1971rate}. We further show that this reduction is related to the difference in differential entropy of the received training signals transmitted over the noisy downlink channel and those over its noiseless counterpart. Non-asymptotic bounds and asymptotic scaling for the MMSE distortion are also provided under finite training, extending the asymptotic results in~\cite{khalilsarai2023fdd,song2025downlink}.

\subsubsection*{Notation}
Bold uppercase letters, e.g., \(\mathbf{H}\), denote random vectors; sans-serif uppercase letters, e.g., \(\mathsf{A}\), denote deterministic matrices; blackboard bold letters, e.g., \(\mathbb{X}\), denote random matrices; and calligraphic letters, e.g., \(\mathcal{S}\), denote sets.
Let \(\tr(\mathsf{A})\), \(\rank(\mathsf{A})\), and \(\lambda_{i}(\mathsf{A})\) denote the trace, rank, and eigenvalues of a matrix, respectively. 
The pseudo-determinant of a Hermitian matrix is given by \(\det^{+}(\mathsf{A}) = \prod_{\lambda_{i}(\mathsf{A})>0} \lambda_{i}(\mathsf{A})\).
We use \(\mathcal{CN}(\cdot, \cdot)\) and \(\mathcal{CW}_{m}(n, \mathsf{\Sigma})\) to denote the circularly symmetric complex Gaussian distribution and the complex Wishart distribution with \(n\) degrees of freedom and covariance matrix \(\mathsf{\Sigma} \in \mathcal{C}^{m \times m}\), respectively.
\(\mathsf{A}_{[:,i]}\) is the \(i\)-th column, and \(\mathsf{A}_{[j,:]}\) is the \(j\)-th row of the matrix \(\mathsf{A}\).
The symbol \(\otimes\) denotes the Kronecker product.

\section{System Model}

\subsection{Downlink Channel Training}

We consider an FDD multi-antenna system, where a BS equipped with \(m_{t}\) transmit antennas serves a single-antenna UE over an orthogonal frequency-division multiplexing (OFDM)-based downlink with \(n_{c}\) subcarriers. The downlink channel is assumed to be quasi-static over at least \(n_{t}\)  consecutive OFDM symbol intervals. For \(n_{p} \leq n_{c}\) pilot-bearing subcarriers, the training symbols received by the UE are given by
\begin{align} \label{eq:DLCSI_training}
\mathbf{Y} = \mathbf{H}^\mathsf{H} \mathbb{P} + \mathbf{W},
\end{align}
where \(\mathbf{W} \sim \mathcal{CN}(\mathbf{0}, \mathsf{I}_{n_{0}})\) is additive white Gaussian noise (AWGN) and \(n_{0} = n_{t} n_{p}\).
The downlink CSI vector \(\mathbf{H}\) represents the channel coefficients stacked across all \(m_t\) antennas and \(n_c\) subcarriers.
The downlink pilot symbol matrix \(\mathbb{P} \in \mathcal{C}^{N \times n_{0}}\) with \(N = m_{t} n_{c}\) is subject to the average transmit power constraint per time-frequency resource element, i.e., \(\mathds{E}[\|\mathbb{P}_{[:,i]}\|^{2}] \leq \mathrm{SNR}_{\mathrm{dl}}\) for all \(i \in \{1, \dots, n_{0}\}\).

To facilitate downlink channel training, in the spirit of \cite{khalilsarai2023fdd,song2025downlink}, let \(\mathcal{N}_{p} \subseteq \{1,\dots,n_{c}\}\) denote the set of \(n_{p}\) subcarriers selected for pilot transmission.  
The downlink pilot matrix \(\mathbb{P}\) is drawn from a Gaussian ensemble and can be formulated as
\begin{align} \label{eq:pilot_matrix}
\mathbb{P} = 
\Big[
\mathbf{e}_{\mathcal{N}_{p}(1)} \otimes \mathbb{P}_{1}, \;
\mathbf{e}_{\mathcal{N}_{p}(2)} \otimes \mathbb{P}_2, \;
\dots, \;
\mathbf{e}_{\mathcal{N}_{p}(n_{p})} \otimes \mathbb{P}_{n_{p}}
\Big],
\end{align}
where \(\mathbf{e}_{\mathcal{N}_{p}(b)} \in \mathcal{R}^{n_{c}}\) is a one-hot vector with 1 at the position \(\mathcal{N}_{p}(b)\) corresponding to the \(b\)-th pilot subcarrier and 0 elsewhere.
Each pilot block \(\mathbb{P}_b \in \mathcal{C}^{m_{t} \times n_{t}}\) consists of columns drawn independently from \(\mathcal{CN}(\mathbf{0}, p \mathsf{I}_{m_{t}})\) with power constraint \(p = \frac{\mathrm{SNR}_{\mathrm{dl}}}{m_{t}}\).

In this work, we consider a downlink channel \(\mathbf{H} \sim \mathcal{CN}(\mathbf{0}, \mathsf{C}_{\mathbf{H}})\) with both spatial and frequency-domain correlation, which remains quasi-static within each block and is independent across blocks~\cite{yang2014quasi,durisi2015short}. Using the eigenvalue decomposition of \(\mathsf{C}_{\mathbf{H}}\), the Hermitian square root of \(\mathsf{C}_{\mathbf{H}}\) is given by
\(\mathsf{C}_{\mathbf{H}}^{1/2}
= \mathsf{U} \mathsf{\Lambda}_{\mathbf{H}}^{1/2} \mathsf{U}^{\mathsf{H}}\). 

\begin{remark}[Availability of channel statistics]
Channel statistics, such as the mean and covariance, generally vary slowly over time compared to the coherence time of the channel. In practical systems, these statistics can be estimated at the BS from uplink measurements via reciprocity, and at the UE from downlink pilots~\cite{dietrich2005pilot, yang2023plug, khalilsarai2023fdd, song2025downlink}. Therefore, in this work, we assume that these statistics are available.
\hfill\ensuremath{\lozenge}
\end{remark}

Given the pilot matrix from Eq.~\eqref{eq:pilot_matrix}, the covariance matrix of the received downlink training signal in Eq.~\eqref{eq:DLCSI_training} at the UE is 
\begin{align}
   \mathsf{C}_{\mathbf{Y} \mid \mathbb{P}} &= \mathds{E}[\mathbf{Y}^{\mathsf{H}}\mathbf{Y} | \mathbb{P}] = \mathbb{A}^{\mathsf{H}}\mathbb{A}+\mathsf{I}_{n_{0}},
\end{align} 
where we define \(\mathbb{A} \triangleq \mathsf{C}_{\mathbf{H}}^{1/2} \mathbb{P}\).
We also introduce a noiseless version of the downlink training, analogous to Eq.~\eqref{eq:DLCSI_training}:  
\(\overline{\mathbf{Y}} = \mathbf{H}^{\mathsf{H}} \mathbb{P}\),  
whose covariance matrix is given by  
\begin{align}
   \mathsf{C}_{\overline{\mathbf{Y}} \mid \mathbb{P}} = \mathbb{A}^{\mathsf{H}}\mathbb{A}.
\end{align}
Recalling that, given \(\mathbb{P}\), both \(\mathbf{Y}\) and \(\overline{\mathbf{Y}}\) follow complex Gaussian distributions, their differential entropies are~\cite{thomas2006elements}
\begin{align}
   h(\mathbf{Y} | \mathbb{P}) &= r_{Y} \log(\pi e) + \log\big(\det^{+}(\mathbb{A}^{\mathsf{H}}\mathbb{A} + \mathsf{I}_{n_{0}})\big), \label{eq:training_signal_entropy} \\
   h(\overline{\mathbf{Y}} | \mathbb{P}) &= r_{\overline{Y}} \log(\pi e) + \log\big(\det^{+}(\mathbb{A}^{\mathsf{H}}\mathbb{A})\big). \label{eq:training_noiseless_signal_entropy}
\end{align}
where \(r_{Y} = \rank(\mathsf{C}_{\mathbf{Y} \mid \mathbb{P}})\) and \(r_{\overline{Y}} = \rank(\mathsf{C}_{\overline{\mathbf{Y}} \mid \mathbb{P}})\).
It can be observed that \(\mathbb{A}\) precisely represents a proper format of the eigenspace of the received signal covariance~\cite[Sec.~III-C]{song2025downlink}.
Moreover, in this work, we show that \(\mathbb{A}\) plays a central role in the theoretical analysis of the CSI feedback overall performance.

\subsection{Channel Estimation and Feedback}

The UE performs channel estimation upon receiving the downlink training transmission of \(n_{t}\) pilot symbols.
Specifically, given the Gaussian pilot \(\mathbb{P}\), the MMSE estimate of \(\mathbf{H}\) based on the observation \(\mathbf{Y}\) is
\begin{align}
\mathbf{S} = \mathds{E}[\mathbf{H} | \mathbf{Y}] 
= \mathsf{C}_{\mathbf{H}} \mathbb{P} \big(\mathbb{P}^{\mathsf{H}} \mathsf{C}_{\mathbf{H}} \mathbb{P} + \mathsf{I}_{n_{0}}\big)^{-1} \mathbf{Y}^{\mathsf{H}}.
\end{align}
The corresponding estimation covariance is
\begin{align}
\mathsf{C}_{\mathbf{S} \mid \mathbb{P}} &= \mathsf{C}_{\mathbf{H}} \mathbb{P} \big(\mathbb{P}^{\mathsf{H}} \mathsf{C}_{\mathbf{H}} \mathbb{P} + \mathsf{I}_{n_{0}}\big)^{-1} \mathbb{P}^{\mathsf{H}} \mathsf{C}_{\mathbf{H}} \\
&= \mathsf{C}_{\mathbf{H}}^{1/2} \mathbb{A}\left(\mathbb{A}^{\mathsf{H}}\mathbb{A}+\mathsf{I}\right)^{-1}\mathbb{A}^{\mathsf{H}} \mathsf{C}_{\mathbf{H}}^{1/2}. \label{eq:MMSE_estimator}
\end{align}
The MMSE distortion resulting from the UE's estimate is
\begin{align}
D_{\mathrm{mmse}|\mathbb{P}} &= \tr(\mathsf{C}_{\mathbf{H}} - \mathsf{C}_{\mathbf{S} \mid \mathbb{P}}).
\end{align}

\begin{figure}[!t]
   \centering
   \includegraphics[width=1\linewidth]{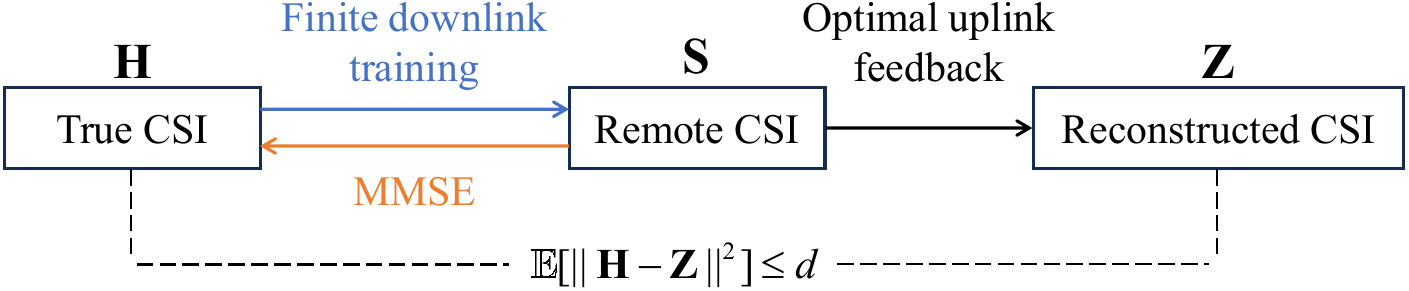}
   \caption{System model for CSI feedback.}
   \label{fig:model}
   \vspace{-0.5em}
\end{figure}

The overall CSI feedback system is represented as the Markov chain \(\mathbf{H} \to \mathbf{S} \to \mathbf{Z}\) in Fig.~\ref{fig:model}. Given the pilot matrix \(\mathbb{P}\), the MMSE estimate \(\mathbf{S}\) serves as an intermediate representation of the true channel \(\mathbf{H}\). 
The UE then quantizes and encodes \(\mathbf{S}\) for uplink feedback, leading to a reconstruction \(\mathbf{Z}\) at the BS. 
The overall distortion caused by the CSI acquisition process is measured by 
\(\mathsf{d}(\mathbf{H}, \mathbf{Z}) = \mathds{E}\left[\|\mathbf{H} - \mathbf{Z}\|^{2}\right]\). 
Due to the presence of the intermediate observation \(\mathbf{S}\), it is worthwhile to recall that the framework of \textit{remote source coding}, introduced in~\cite{khalilsarai2023fdd,song2025downlink}, is used to characterize the information-theoretic limit of CSI acquisition via downlink training and uplink feedback.
The indirect RDF of \(\mathbf{H}\) observed through \(\mathbf{S}\) is defined as~\cite{kostina2016nonasymptotic}
\begin{align} \label{eq:RDF_original_form}
R_{\mathbf{H},\mathbf{S}|\mathbb{P}}(d) 
\triangleq 
\inf_{
\substack{
P_{\mathbf{Z}|\mathbf{S}}:
\mathds{E}[\mathsf{d}(\mathbf{H}, \mathbf{Z})] \leq d
}
}
I(\mathbf{Z}; \mathbf{S}|\mathbb{P}).
\end{align}
In this paper, we explicitly provide the RDF of the entire CSI feedback process. The relationship between the RDF and the number of downlink training transmissions \(n_{t}\) is then analyzed.

\section{Main Results}

\subsection{Derivation of RDF}

In the overall feedback system, the true source is the downlink CSI \(\mathbf{H}\), and the corresponding distortion constraint imposed on it is \(d\). However, the true source \(\mathbf{H}\) is not fully accessible at the UE. The UE observes only the MMSE estimate \(\mathbf{S}\) given the pilot matrix \(\mathbb{P}\), and the effective allowable distortion for compressing \(\mathbf{S}\) becomes \(d - D_{\mathrm{mmse}|\mathbb{P}}\). This inevitably causes the RDF to depend on the pilot matrix \(\mathbb{P}\). To quantify this dependence, we decompose the RDF into two parts. The first part is independent of \(\mathbb{P}\). The second part captures the random deviations in the RDF arising from the limited downlink training with pilots drawn from a Gaussian ensemble. In the second part, we explicitly distinguish between the deviation caused by compressing \(\mathbf{S}\) given \(\mathbb{P}\) instead of \(\mathbf{H}\), and the deviation caused by using an effective distortion of \(d - D_{\mathrm{mmse}|\mathbb{P}}\) instead of \(d\). 
These observations are summarized in the following theorem.

\begin{theorem}
\label{thm:RDF_main}
Consider a downlink channel vector 
\(\mathbf{H} \sim \mathcal{CN}(\mathbf{0}, \mathsf{C}_{\mathbf{H}})\)
with arbitrary spatial correlation among \(m_{t}\) antennas and frequency correlation across \(n_{c}\) subcarriers.  
Using the Gaussian ensemble pilot matrix \(\mathbb{P}\) over \(n_{t}\) training transmissions on subcarriers \(\mathcal{N}_{p} \subseteq \{1,\dots,n_{c}\}\),  
the UE obtains the MMSE estimate \(\mathbf{S}\), which is encoded and fed back for reconstruction \(\mathbf{Z}\).  
Under the distortion constraint 
\(\mathds{E}[\|\mathbf{H}-\mathbf{Z}\|^{2}] \leq d\),
the corresponding RDF is
\begin{align} \label{eq:RDF_v1}
R_{\mathbf{H},\mathbf{S}|\mathbb{P}}(d) = R_{\mathbf{H}}(d) + \Delta R_{\mathbb{P},\mathbf{S}} + \Delta R_{\mathbb{P},d},
\end{align} 
where
\begin{align}
R_{\mathbf{H}}(d) &= \frac{1}{2}\log\big(\det^{+}(\mathsf{C}_{\mathbf{H}})\big) - \frac{r_{H}}{2}\log\frac{d}{r_{H}}, \label{eq:RDF_v1_direct} \\
\Delta R_{\mathbb{P},\mathbf{S}} &= \frac{1}{2} \big(h(\mathbf{S} | \mathbb{P}) - h(\mathbf{H})\big) + \frac{r_{H}}{2}\log\frac{r_{S}}{r_{H}}, \label{eq:RDF_v1_source} \\
\Delta R_{\mathbb{P},d} &= \frac{r_{H}-r_{S}}{2}\log\frac{d \pi e}{r_{S}} - \frac{r_{S}}{2}\log \Big(1 - \frac{D_{\mathrm{mmse}|\mathbb{P}}}{d} \Big), \label{eq:RDF_v1_distortion}
\end{align} 
and
\begin{align} 
   h(\mathbf{H}) &= r_{H} \log(\pi e) + \log\big(\det^{+}(\mathsf{C}_{\mathbf{H}})\big), \label{eq:channel_entropy} \\
   h(\mathbf{S} | \mathbb{P}) &= r_{S} \log(\pi e) + \log\big(\det^{+}\left(\mathsf{C}_{\mathbf{H}}(\mathsf{I}+\mathbb{A}\mathbb{A}^{\mathsf{H}})^{-1}\mathbb{A}\mathbb{A}^{\mathsf{H}}\right)\big), \label{eq:channel_mmse_entropy} \\
   D_{\mathrm{mmse}|\mathbb{P}} &= \tr\big(\mathsf{C}_{\mathbf{H}} (\mathsf{I} + \mathbb{A} \mathbb{A}^{\mathsf{H}})^{-1}\big). \label{eq:Dmmse_trace_form}
\end{align} 
Here we let \(\mathbb{A} = \mathsf{C}_{\mathbf{H}}^{1/2}\mathbb{P}\), \(r_{H} = \rank(\mathsf{C}_{\mathbf{H}})\), \(r_{S} = \rank(\mathsf{C}_{\mathbf{S} \mid \mathbb{P}})\) and assume that the distortion constraint \(d\) satisfies
\begin{align} \label{eq:distortion_range}
   \tr(\mathsf{C}_{\mathbf{H}}) - \tr(\mathsf{C}_{\mathbf{S} \mid \mathbb{P}}) &\leq d \leq \tr(\mathsf{C}_{\mathbf{H}}) - \tr(\mathsf{C}_{\mathbf{S} \mid \mathbb{P}}) + r_{S}\lambda_{\min}(\mathsf{C}_{\mathbf{S} \mid \mathbb{P}}), 
\end{align} 
where \(\mathsf{C}_{\mathbf{S} \mid \mathbb{P}} = \mathsf{C}_{\mathbf{H}} \mathbb{P} \big(\mathbb{P}^{\mathsf{H}} \mathsf{C}_{\mathbf{H}} \mathbb{P} + \mathsf{I}_{n_{0}}\big)^{-1} \mathbb{P}^{\mathsf{H}} \mathsf{C}_{\mathbf{H}}\).
\end{theorem}

\begin{IEEEproof}
See Appendicx~\ref{proof:thm:RDF_main}.
\end{IEEEproof}

\begin{remark}[Range of the imposed distortion constraint]
The range of \(d\) in Eq.~\eqref{eq:distortion_range} specifies the minimal distortion where the closed-form RDF applies. More generally, if \(d\) satisfies
\begin{align}
r'_{S} \, \lambda_{i}(\mathsf{C}_{\mathbf{S} \mid \mathbb{P}}) 
\leq d - \big(\tr(\mathsf{C}_{\mathbf{H}}) - \tr(\mathsf{C}_{\mathbf{S} \mid \mathbb{P}})\big) 
\leq r'_{S} \, \lambda_{i+1}(\mathsf{C}_{\mathbf{S} \mid \mathbb{P}}),
\end{align}
where \(\lambda_{i}(\mathsf{C}_{\mathbf{S} \mid \mathbb{P}})\) denotes the \(i\)-th smallest nonzero eigenvalue of \(\mathsf{C}_{\mathbf{S} \mid \mathbb{P}}\) and \(r'_{S} = r_{S} - i\). Thus, Theorem~\ref{thm:RDF_main} still holds if \(r_{S}\) in Eq.~\eqref{eq:RDF_v1} is replaced by \(r'_{S}\).
In practice, since the distortion introduced by CSI feedback is typically small~\cite{khalilsarai2023fdd,song2025downlink}, it suffices to consider the most stringent distortion range.
\hfill\ensuremath{\lozenge}
\end{remark}

In this theorem, we derive the overall RDF for the CSI feedback system, where the UE estimates the downlink channel via MMSE after \(n_{t}\) training transmissions with Gaussian pilots and subsequently encodes the estimated CSI for feedback. Specifically, the overall RDF \(R_{\mathbf{H},\mathbf{S}|\mathbb{P}}(d)\) in Eq.~\eqref{eq:RDF_v1} is decomposed into three components. The first term, \(R_{\mathbf{H}}(d)\) in Eq.~\eqref{eq:RDF_v1_direct}, can be interpreted as the RDF for directly encoding the true source \(\mathbf{H}\) under distortion \(d\). The second term, \(\Delta R_{\mathbb{P},\mathbf{S}}\) in Eq.~\eqref{eq:RDF_v1_source}, quantifies the decrease in compression cost when encoding the MMSE estimate \(\mathbf{S}\) given \(\mathbb{P}\), instead of the true source \(\mathbf{H}\), as the MMSE estimation reduces the source uncertainty. The third term, \(\Delta R_{\mathbb{P},d}\) in Eq.~\eqref{eq:RDF_v1_distortion}, quantifies the increase in compression cost resulting from the smaller effective distortion allowance \(d - D_{\mathrm{mmse}|\mathbb{P}}\).
It is worth emphasizing that both \(\Delta R_{\mathbb{P},\mathbf{S}}\) and \(\Delta R_{\mathbb{P},d}\) can cause the remote RDF \(R_{\mathbf{H},\mathbf{S}|\mathbb{P}}(d)\) to deviate from the direct RDF \(R_{\mathbf{H}}(d)\), primarily due to the strong randomness introduced by the Gaussian pilot matrix under a limited number of training transmissions.

In the following, to better quantify the impact of pilot-induced randomness on overall feedback performance, we examine how the training length \(n_{t}\) influences the expected feedback efficiency (the established RDFs) under common and mild conditions.

\subsection{RDF Decrease from Compressing the MMSE Estimate}

We assume that \(\mathsf{C}_{\mathbf{H}}\) is full rank. For analytical tractability, we consider the case where pilots occupy all subcarriers (\(n_{c} = n_{p}\)) and focus on the regime \(n_{t} \geq m_{t}\) to ensure that the estimation in Eq.~\eqref{eq:DLCSI_training} is well-defined.
This setting is practically relevant, as it provides sufficient observations for accurate channel estimation and supports reliable overall feedback performance.
It is well known that, under the same distortion constraint, compressing the MMSE estimate 
\(\mathbf{S} = \mathds{E}[\mathbf{H} | \mathbf{Y}]\) incurs a lower rate than compressing \(\mathbf{H}\) directly. 
This is because \(\mathbf{S}\) contains only the components of \(\mathbf{H}\) that are correlated with \(\mathbf{Y}\), while the remaining uncorrelated part arises from the AWGN \(\mathbf{W}\) in \(\mathbf{Y}\). This observation can be formally stated as follows.

\begin{proposition} \label{prop:RDF_decrease}
When \(n_{c} = n_{p}\) and \(n_{t} \geq m_{t}\), under the same distortion constraint and a given pilot matrix \(\mathbb{P}\) drawn from a Gaussian ensemble, the RDF reduction from compressing the MMSE estimate \(\mathbf{S} = \mathds{E}[\mathbf{H} | \mathbf{Y}]\) of \(\mathbf{H}\), rather than \(\mathbf{H}\), is
\begin{align} \label{eq:RDF_v1_source_entropy}
\Delta R_{\mathbb{P},\mathbf{S}} &= \frac{1}{2} \big(h(\mathbf{S} | \mathbb{P}) - h(\mathbf{H})\big) \\
&= \frac{1}{2} \big(h(\overline{\mathbf{Y}} | \mathbb{P}) - h(\mathbf{Y} | \mathbb{P})\big),
\end{align}
where the four differential entropies are given in Eqs.~\eqref{eq:training_signal_entropy}, \eqref{eq:training_noiseless_signal_entropy}, \eqref{eq:channel_entropy}, \eqref{eq:channel_mmse_entropy}, respectively.
\end{proposition}
\begin{IEEEproof}
   Given that \(\mathbb{P}\) is generated according to Eq.~\eqref{eq:pilot_matrix} with \(n_{0} \geq N = m_{t} n_{c}\), it is almost surely of full rank \(N\). Thus, \(\mathsf{C}_{\mathbf{S} \mid \mathbb{P}}\) has rank \(r_{S} = N\) almost surely, and \(\mathbb{A}\mathbb{A}^{\mathsf{H}}\) is also full rank, where \(\mathbb{A} = \mathsf{C}_{\mathbf{H}}^{1/2} \mathbb{P}\). In this case, \(\Delta R_{\mathbb{P},\mathbf{S}}\) in Eq.~\eqref{eq:RDF_v1_source} can be elegantly expressed by referring to Eqs.~\eqref{eq:channel_entropy} and \eqref{eq:channel_mmse_entropy} as
\begin{align}
\Delta R_{\mathbb{P},\mathbf{S}} = \frac{1}{2}\log \frac{\det\left(\mathbb{A}\mathbb{A}^{\mathsf{H}}\right)}{\det\left(\mathsf{I}+\mathbb{A}\mathbb{A}^{\mathsf{H}}\right)}.
\end{align}
Recalling Eqs.~\eqref{eq:training_signal_entropy} and  \eqref{eq:training_noiseless_signal_entropy}, and noting that \(\mathbb{A}\mathbb{A}^{\mathsf{H}}\) and \(\mathbb{A}^{\mathsf{H}}\mathbb{A}\) share the same nonzero eigenvalues, the result follows.
\end{IEEEproof}

It can be observed that the MMSE estimator provides a theoretically efficient approach for CSI feedback. Under the same distortion constraint, the MMSE-based compression achieves a reduction in the RDF, which quantifies the fraction of the received information rendered irrelevant to the true CSI due to the presence of downlink noise.
Next, we present the relationship between \(\Delta R_{\mathbb{P},\mathbf{S}}\) and the downlink training length \(n_{t}\) in expectation.

\begin{theorem} \label{thm:RDF_decrease}
When \(n_{c} = n_{p}\) and \(n_{t} \geq m_{t}\), under the same distortion constraint, the expected RDF reduction \(\mathds{E}_{\mathbb{P}}[\Delta R_{\mathbb{P},\mathbf{S}}]\) from compressing the MMSE estimate \(\mathbf{S}\), with the pilot matrix \(\mathbb{P}\) drawn from a Gaussian ensemble in Eq.~\eqref{eq:pilot_matrix}, is bounded as
\begin{align}
   \Delta R_{\mathbf{S}}^{L} \leq \mathds{E}_{\mathbb{P}}[\Delta R_{\mathbb{P},\mathbf{S}}] \leq \Delta R_{\mathbf{S}}^{U},
\end{align}
where \(\Delta R_{\mathbb{P},\mathbf{S}}\) is in Eq.~\eqref{eq:RDF_v1_source}, and
\begin{align}
   &\Delta R_{\mathbf{S}}^{L} = \frac{1}{2} \Big(\log \frac{\det(\mathsf{C}_{\mathbf{H}})}{\det\big(\mathsf{I}+\mathrm{SNR}_{\mathrm{dl}} \frac{n_{t}}{m_{t}}\,\mathsf{C}_{\mathbf{H}}\big)} + m_{t} n_{c} \log \frac{\mathrm{SNR}_{\mathrm{dl}}}{m_{t}} \notag \\
   &\quad \quad \quad \; + n_{c} \sum_{j=0}^{m_{t}-1} \psi(n_{t} - j)\Big), \label{eq:RDF_source_LB} \\
   &\Delta R_{\mathbf{S}}^{U} = \frac{m_{t}}{2\mathrm{SNR}_{\mathrm{dl}} (n_{t} - m_{t})} \cdot \Big(-\tr\big(\mathsf{C}_{\mathbf{H}}^{-1}\big) \notag \\
   & \quad + \frac{m_{t}^{2}n_{c}n_{t}}{2\mathrm{SNR}_{\mathrm{dl}}(n_{t}-m_{t}-1)(n_{t}-m_{t}+1)\lambda_{\min}^{2}(\mathsf{C}_{\mathbf{H}})}\Big). \label{eq:RDF_source_UB}
\end{align}
Furthermore, in the regime of large \(n_{t}\), we have the asymptotic scaling
\begin{align} \label{eq:RDF_source_asymptotics}
   \mathds{E}_{\mathbb{P}}[\Delta R_{\mathbb{P},\mathbf{S}}] = \frac{C_{\mathbf{S}}}{n_{t}} + O\!\left(\frac{1}{n_{t}^{2}}\right),
\end{align}
where \(C_{\mathbf{S}}^{L} \leq C_{\mathbf{S}} \leq C_{\mathbf{S}}^{U}\), and 
\begin{align}
   C_{\mathbf{S}}^{L} &= - \frac{m_{t}\tr(\mathsf{C}_{\mathbf{H}}^{-1})}{2\mathrm{SNR}_{\mathrm{dl}}} - \frac{m_{t}^{2}n_{c}}{4}, \label{eq:RDF_source_LB_coefficient} \\
   C_{\mathbf{S}}^{U} &= - \frac{m_{t}\tr(\mathsf{C}_{\mathbf{H}}^{-1})}{2\mathrm{SNR}_{\mathrm{dl}}}. \label{eq:RDF_source_UB_coefficient}
\end{align}
\end{theorem}

\begin{IEEEproof}
   See Appendicx~\ref{proof:thm:RDF_decrease}.
\end{IEEEproof}

In this theorem, it is shown that as the number of downlink training symbols \(n_{t}\) increases, the RDF decrease from compressing the MMSE estimate decays towards zero at the rate of \(\frac{1}{n_{t}}\).
Moreover, by combining this result with Proposition~\ref{prop:RDF_decrease}, it follows that, even when the downlink SNR remains fixed, the noisy downlink channel progressively resembles its noiseless counterpart as the training length increases.
\subsection{RDF Increase from a Stricter Distortion Constraint}

In this part, the key observation is that the UE performs an MMSE estimation, which effectively reduces the allowable uplink distortion to \(d - D_{\mathrm{mmse}|\mathbb{P}}\).
This reduction leads to an increase in the corresponding RDF. The related results are summarized in the following theorem.

\begin{theorem} \label{thm:RDF_increase}
When \(n_{c} = n_{p}\) and \(n_{t} \geq m_{t}\), if the overall distortion constraint \(d\) is imposed as specified in Eq.~\eqref{eq:distortion_range}, then with the pilot matrix \(\mathbb{P}\) drawn from the Gaussian ensemble in Eq.~\eqref{eq:pilot_matrix}, the expected MMSE at the UE side is bounded as  
\begin{align}
D_{\mathrm{mmse}}^{L} \leq \mathds{E}_{\mathbb{P}}[D_{\mathrm{mmse}|\mathbb{P}}] \leq D_{\mathrm{mmse}}^{U},
\end{align}
where \(D_{\mathrm{mmse}|\mathbb{P}}\) is in Eq.~\eqref{eq:Dmmse_trace_form}, and
\begin{align}
   &D_{\mathrm{mmse}}^{L} = \tr\Big(\mathsf{C}_{\mathbf{H}} \Big(\mathsf{I} + \mathrm{SNR}_{\mathrm{dl}} \frac{n_{t}}{m_{t}} \mathsf{C}_{\mathbf{H}}\Big)^{-1}\Big), \label{eq:Dmmse_LB} \\
   &D_{\mathrm{mmse}}^{U} = \frac{m_{t}^{2} n_{c}}{\mathrm{SNR}_{\mathrm{dl}}(n_{t}-m_{t})}. \label{eq:Dmmse_UB}
\end{align}
The expected MMSE exhibits the scaling
\begin{align} \label{eq:Dmmse_asymptotics}
\mathds{E}_{\mathbb{P}}[D_{\mathrm{mmse}|\mathbb{P}}] = \frac{m_{t}^{2} n_{c}}{\mathrm{SNR}_{\mathrm{dl}}} \cdot \frac{1}{n_{t}} + O\!\left(\frac{1}{n_{t}^{2}}\right).
\end{align}

Under the reduced distortion constraint \(d - D_{\mathrm{mmse}|\mathbb{P}}\), the corresponding expected increase in the RDF is bounded as  
\begin{align}
   \Delta R_{d}^{L} \leq \mathds{E}_{\mathbb{P}}[\Delta R_{\mathbb{P},d}] \leq \Delta R_{d}^{U},
\end{align}
where \(\Delta R_{\mathbb{P},d}\) is in Eq.~\eqref{eq:RDF_v1_distortion}, and
\begin{align}
   &\Delta R_{d}^{L} = -\frac{m_{t} n_{c}}{2}\,\log\Big(1-\frac{D_{\mathrm{mmse}}^{L}}{d}\Big), \label{eq:RDF_distortion_LB} \\
   &\Delta R_{d}^{U} = \frac{m_{t} n_{c}}{2\epsilon d}\,D_{\mathrm{mmse}}^{U}, \label{eq:RDF_distortion_UB}
\end{align}
for some given \(\epsilon \in (0,1)\). 
Furthermore, in the regime of large \(n_{t}\), we have the asymptotic scaling
\begin{align}
\mathds{E}_{\mathbb{P}}[\Delta R_{\mathbb{P},d}] = \frac{C_{d}}{n_{t}} + O\!\left(\frac{1}{n_{t}^{2}}\right),
\end{align}
where \(C_{d}^{L} \leq C_{d} \leq C_{d}^{U}\), and
\begin{align}
   C_{d}^{L} &= \frac{m_{t}^{3}n_{c}^{2}}{2d\,\mathrm{SNR}_{\mathrm{dl}}}, \label{eq:RDF_distortion_LB_coefficient} \\
   C_{d}^{U} &= \frac{C_{d}^{L}}{\epsilon}. \label{eq:RDF_distortion_UB_coefficient}
\end{align}
\end{theorem}

\begin{IEEEproof}
   See Appendicx~\ref{proof:thm:RDF_increase}.
\end{IEEEproof}

Both the additional RDF caused by the stricter distortion constraint and the MMSE distortion vanish at a rate of \(\frac{1}{n_{t}}\).
Moreover, the asymptotic analysis of the MMSE distortion in Eq.~\eqref{eq:Dmmse_asymptotics} shows that \(D_{\mathrm{mmse}}\) decays inversely with the downlink SNR.
This recovers the result in~\cite{khalilsarai2023fdd} while extending it to hold for arbitrary SNR.

\subsection{Impact of Finite Training on the Overall Performance}

We are now ready to characterize, in a statistical sense, the relationship between the overall RDF of the CSI feedback system and the number of downlink training symbols.

\begin{figure}[!t]
   \centering
   \includegraphics[width=0.85\linewidth]{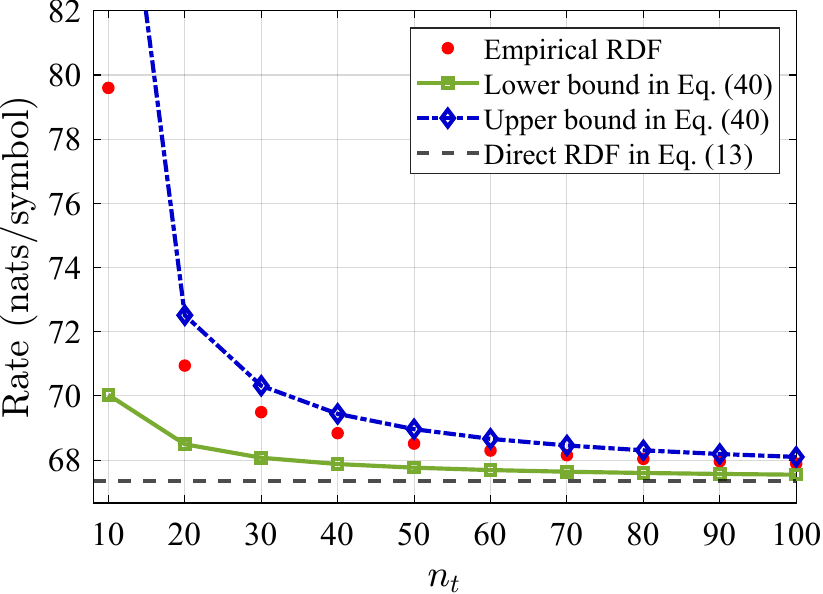}
   \caption{Relationship between downlink training length and the overall RDF of the CSI feedback system.}
   \label{fig:simulation}
   \vspace{-0.5em}
\end{figure}

\begin{theorem}
When \(n_{c} = n_{p}\) and the BS trains the downlink channel using the pilot matrix \(\mathbb{P}\) drawn from the Gaussian ensemble in Eq.~\eqref{eq:pilot_matrix} for \(n_{t}\) symbols with \(n_{t} \geq m_{t}\), and the UE performs an MMSE estimation of \(\mathbf{H}\) before encoding the feedback, the overall expected RDF under the distortion constraint \(\mathds{E}[\|\mathbf{H}-\mathbf{Z}\|^{2}] \leq d\) is bounded as
\begin{align}
\Delta R_{\mathbf{S}}^{L} + \Delta R_{d}^{L} \leq \mathds{E}_{\mathbb{P}}[R_{\mathbf{H},\mathbf{S}|\mathbb{P}}(d)] - R_{\mathbf{H}}(d) \leq \Delta R_{\mathbf{S}}^{U} + \Delta R_{d}^{U},
\end{align}
where the terms \(R_{\mathbf{H}}(d)\), \(\Delta R_{\mathbf{S}}^{L}, \Delta R_{\mathbf{S}}^{U}, \Delta R_{d}^{L}, \Delta R_{d}^{U}\) are given in Eqs.~\eqref{eq:RDF_v1_direct}, \eqref{eq:RDF_source_LB}, \eqref{eq:RDF_source_UB}, \eqref{eq:RDF_distortion_LB}, and \eqref{eq:RDF_distortion_UB}, respectively.

Furthermore, in the regime of large \(n_{t}\), the asymptotic scaling is
\begin{align}
\mathds{E}_{\mathbb{P}}[R_{\mathbf{H},\mathbf{S}|\mathbb{P}}(d)] = R_{\mathbf{H}}(d) + \frac{C_{\mathbf{S},d}}{n_{t}} + O\!\left(\frac{1}{n_{t}^{2}}\right),
\end{align}
where
\begin{align}
C_{\mathbf{S}}^{L} + C_{d}^{L} \leq C_{\mathbf{S},d} \leq C_{\mathbf{S}}^{U} + C_{d}^{U},
\end{align}
with each term provided in Eqs.~\eqref{eq:RDF_source_LB_coefficient}, \eqref{eq:RDF_source_UB_coefficient}, \eqref{eq:RDF_distortion_LB_coefficient}, and \eqref{eq:RDF_distortion_UB_coefficient}.
\end{theorem}
\begin{IEEEproof}
The proof follows directly by combining Theorem~\ref{thm:RDF_decrease} and Theorem~\ref{thm:RDF_increase}.
\end{IEEEproof}
As the number of downlink training pilots increases, the overall RDF based on the UE's MMSE estimate asymptotically converges the direct RDF, corresponding to the case where the UE has full access to the downlink CSI. The gap between the two decays inversely with the number of training symbols.

For the simulations, we consider a system with \(m_{t} = 4\) transmit antennas, \(n_{c} = 8\) subcarriers, \(n_{p} = 8\) pilot blocks, \(\mathrm{SNR}_{\mathrm{dl}} = 10\,\text{dB}\), and a distortion constraint \(d = 3.5\). To emulate a realistic channel environment, the channel covariance matrix \(\mathsf{C}_{\mathbf{H}}\) is chosen such that roughly one-third of its eigenvalues are small, one-third are large, and the rest are moderate.
The simulation results are presented in Fig.~\ref{fig:simulation}, showing that the derived upper and lower bounds are tight and both decay as \(\frac{1}{n_{t}}\) toward the direct RDF. This confirms that increasing the number of downlink training symbols significantly improves the overall CSI feedback performance. Notably, even when the training length is several times the number of antennas, the overall RDF still exhibits a noticeable deviation from the direct RDF, indicating that appropriate design of training and feedback strategies is essential in practical systems.

\section{Conclusion}

In this work, we have considered a CSI feedback system where the UE performs MMSE estimation before encoding and feeding back to the BS. We provided the RDF of the overall system when the BS employs finite downlink training with Gaussian ensemble pilots. Specifically, the RDF was decomposed into three contributions: the direct RDF corresponding to the UE having access to the downlink CSI, the decrease in RDF due to MMSE lowering the source uncertainty, and the increase in RDF caused by the stricter distortion constraint induced by MMSE in the feedback process. 
A key observation is that as the number of downlink training symbols increases, the overall RDF decays proportionally to the inverse of the training length, eventually converging the RDF of direct CSI encoding. This result is practically relevant, as it indicates that sufficient downlink training allows the UE to effectively behave as if it has full access to the downlink CSI for subsequent feedback. Future work will investigate the system performance when training is performed on only a few pilot subcarriers, under scenarios with extremely limited training and a very large number of BS antennas.

\appendices

\section{
Proof of Theorem~\ref{thm:RDF_main} \label{proof:thm:RDF_main}
}
At the UE, the MMSE estimation error \(\mathbf{E} = \mathbf{H} - \mathbf{S}\) is independent of \(\mathbf{S}\), with \(\mathbf{E} \sim \mathcal{CN}(\mathbf{0}, \mathsf{C}_{\mathbf{E}})\) and \(\mathsf{C}_{\mathbf{E}} = \mathsf{C}_{\mathbf{H}} - \mathsf{C}_{\mathbf{S} \mid \mathbb{P}}\).  
The optimal test channel achieving the RDF is modeled as an AWGN channel \(\mathbf{S} = \mathbf{Z} + \mathbf{N}\), where \(\mathbf{N} \sim \mathcal{CN}(\mathbf{0}, \mathsf{C}_{\mathbf{N}})\) is independent of \(\mathbf{Z}\).
Thus, the distortion becomes \(\mathds{E}[\|\mathbf{H} - \mathbf{Z}\|^{2}] = \mathds{E}[\|\mathbf{E} + \mathbf{N}\|^2]\), and the mutual information \(I(\mathbf{S}; \mathbf{Z})\) can be reformulated accordingly.  
For the Gaussian case, the RDF in Eq.~\eqref{eq:RDF_original_form} is obtained by solving
\begin{align}
   \min_{\mathsf{C}_{\mathbf{N}}} \quad &\frac{1}{2} \log \det^{+} \left( \mathsf{C}_{\mathbf{S} \mid \mathbb{P}} \mathsf{C}_{\mathbf{N}}^{-1} \right), \label{eq:mutual_information}
\\
   \text{s.t.} \quad & \tr(\mathsf{C}_{\mathbf{H}}) - \tr(\mathsf{C}_{\mathbf{S} \mid \mathbb{P}}) + \tr(\mathsf{C}_{\mathbf{N}}) \leq d.
\end{align} 
For rigor, the pseudo-determinant is used to account for the possible rank deficiency of \(\mathsf{C}_{\mathbf{S} \mid \mathbb{P}}\) in Eq.~\eqref{eq:mutual_information}.
This is a trace-constrained maximum-determinant problem~\cite{boyd2004convex}.  
According to the classical reverse water-filling principle~\cite{thomas2006elements}, by allocating the same noise power across all nonzero eigenvalues of \(\mathsf{C}_{\mathbf{S} \mid \mathbb{P}}\), in view of Eq.~\eqref{eq:distortion_range}, the optimal noise covariance is
\begin{align} \label{eq:CN_{c}ov}
   \mathsf{C}_{\mathbf{N}}^{\star} = \frac{d - (\tr(\mathsf{C}_{\mathbf{H}}) - \tr(\mathsf{C}_{\mathbf{S} \mid \mathbb{P}}))}{r_{S}} \, \mathsf{I},
\end{align} 
where \(r_{S}\) is the effective dimension of \(\mathbf{S}\).
Substituting \(\mathsf{C}_{\mathbf{N}}^{\star}\) from Eq.~\eqref{eq:CN_{c}ov} into Eq.~\eqref{eq:mutual_information}, the RDF can be expressed as
\begin{align} \label{eq:RDF_origin_v1}
   R_{\mathbf{H},\mathbf{S}|\mathbb{P}}(d) = \tfrac{1}{2} \log \det^{+} \!\! \left( \mathsf{C}_{\mathbf{S} \mid \mathbb{P}} \right) - \tfrac{r_{S}}{2} \log \! \big( \tfrac{d-(\tr(\mathsf{C}_{\mathbf{H}}) - \tr(\mathsf{C}_{\mathbf{S} \mid \mathbb{P}}))}{r_{S}} \big).
\end{align}
The terms in Eq.~\eqref{eq:RDF_v1_direct} that are independent of \(\mathbb{P}\) are separated from the RDF in Eq.~\eqref{eq:RDF_origin_v1}. 
By defining \(\mathbb{A} = \mathsf{C}_{\mathbf{H}}^{1/2} \mathbb{P}\), we write
\(\mathsf{C}_{\mathbf{S} \mid \mathbb{P}} = \mathsf{C}_{\mathbf{H}}^{1/2} \mathbb{A} (\mathbb{A}^{\mathsf{H}}\mathbb{A} + \mathsf{I})^{-1} \mathbb{A}^{\mathsf{H}} \mathsf{C}_{\mathbf{H}}^{1/2}\)
in Eq.~\eqref{eq:MMSE_estimator}.
Recalling the Woodbury identity,
\begin{align}
   \mathbb{A}(\mathbb{A}^{\mathsf{H}}\mathbb{A} + \mathsf{I})^{-1}\mathbb{A}^{\mathsf{H}} &= (\mathsf{I} + \mathbb{A}\mathbb{A}^{\mathsf{H}})^{-1} \mathbb{A}\mathbb{A}^{\mathsf{H}} \label{eq:Woodbury_identity_part1} \\
   &= \mathsf{I} - (\mathsf{I} + \mathbb{A}\mathbb{A}^{\mathsf{H}})^{-1}, \label{eq:Woodbury_identity_part2}
\end{align}
using Eq.~\eqref{eq:Woodbury_identity_part1}, we obtain \(\Delta R_{\mathbb{P},\mathbf{S}}\) in Eq.~\eqref{eq:RDF_v1_source}. On the other hand, by rewriting \(\mathsf{C}_{\mathbf{S} \mid \mathbb{P}}\) using Eq.~\eqref{eq:Woodbury_identity_part2}, we have \(\Delta R_{\mathbb{P},d}\) in Eq.~\eqref{eq:RDF_v1_distortion} and \(D_{\mathrm{mmse} \mid \mathbb{P}}\) in Eq.~\eqref{eq:Dmmse_trace_form}.

\section{
Proof of Theorem~\ref{thm:RDF_decrease} \label{proof:thm:RDF_decrease}
}
We have \(\mathds{E}_{\mathbb{P}}[\Delta R_{\mathbb{P},\mathbf{S}}] = \frac{1}{2} \big( \mathds{E}[\log \det(\mathbb{A}\mathbb{A}^{\mathsf{H}})] - \mathds{E}[\log \det(\mathsf{I}+\mathbb{A}\mathbb{A}^{\mathsf{H}})] \big)\).
To establish a bound on this quantity, we first derive the exact 
expression of \(\mathds{E}[\log \det(\mathbb{A}\mathbb{A}^{\mathsf{H}})]\), and then proceed to bound \(\mathds{E}[\log \det(\mathsf{I}+\mathbb{A}\mathbb{A}^{\mathsf{H}})]\).
With \(\mathbb{A} = \mathsf{C}_{\mathbf{H}}^{1/2}\mathbb{P}\), we write \(\det(\mathbb{A}\mathbb{A}^{\mathsf{H}}) = \det(\mathsf{C}_{\mathbf{H}})\det(\mathbb{P}\mathbb{P}^{\mathsf{H}})\). Therefore, it suffices to focus on \(\mathds{E}[\log \det(\mathbb{P}\mathbb{P}^{\mathsf{H}})]\).
The pilot matrix \(\mathbb{P}\) is block-diagonal with blocks \(\mathbb{P}_{k} \in \mathcal{C}^{m_{t} \times n_{t}}\), where each block \(\mathbb{P}_{k}\mathbb{P}_{k}^{\mathsf{H}}\) follows a complex Wishart distribution \(\mathcal{CW}_{m_{t}}(n_{t}, p\mathsf{I})\) when \(n_{t} \geq m_{t}\).
For \(\mathbb{X} \sim \mathcal{CW}_{m}(n, \mathsf{\Sigma})\) with \(n \geq m\), we have~\cite[Eq. (2.12)]{tulino2004random} 
\begin{align}
\mathds{E}[\log \det(\mathbb{X})] = \log \det(\mathsf{\Sigma}) + \sum\nolimits_{j=0}^{m-1} \psi(n-j),
\end{align}
where \(\psi(\cdot)\) is the digamma function. Using this result, we have \(\mathds{E}[\log \det(\mathbb{P}_{k}\mathbb{P}_{k}^{\mathsf{H}})] = m_{t} \log p + \sum_{j=0}^{m_{t}-1} \psi(n_{t} - j)\), where \(p = \frac{\mathrm{SNR}_{\mathrm{dl}}}{m_{t}}\) denotes the power constraint of the pilot matrix \(\mathbb{P}\).
Thus summing over all blocks yields  
\begin{align}
&\mathds{E}[\log \det(\mathbb{A}\mathbb{A}^{\mathsf{H}})] = \log \det(\mathsf{C}_{\mathbf{H}}) \notag \\
&\quad \quad \quad \quad \quad \quad + n_{c} \Big(m_{t} \log p + \sum\nolimits_{j=0}^{m_{t}-1} \psi(n_{t} - j)\Big). \label{eq:E_log_det_AAH}
\end{align}
Applying Jensen's inequality, we have
\begin{align}
\mathds{E}\big[\log \det(\mathsf{I}+\mathbb{A}\mathbb{A}^{\mathsf{H}})\big]
&\leq \log \det\big(\mathsf{I}+\mathds{E}[\mathbb{A}\mathbb{A}^{\mathsf{H}}]\big) \label{eq:E_log_det_OnePlusAAH} \\
&= \log \det\big(\mathsf{I}+\mathsf{C}_{\mathbf{H}}^{1/2}\,\mathds{E}[\mathbb{P}\mathbb{P}^{\mathsf{H}}]\,\mathsf{C}_{\mathbf{H}}^{1/2}\big) \\
&= \log \det\big(\mathsf{I}+n_{t} p\,\mathsf{C}_{\mathbf{H}}\big), \label{eq:E_log_det_OnePlusAAH_UB}
\end{align}
where the last equality follows from the fact that each block \(\mathbb{P}_{k}\) of the pilot matrix \(\mathbb{P}\) satisfies \(\mathds{E}[\mathbb{P}_{k} \mathbb{P}_{k}^{\mathsf{H}}] = n_{t} p\,\mathsf{I}_{m_{t}}\), so that \(\mathds{E}[\mathbb{P}\mathbb{P}^{\mathsf{H}}]=n_{t} p\,\mathsf{I}_{N}\) with \(N = m_{t}n_{c}\).
Combining Eq.~\eqref{eq:E_log_det_AAH} and Eq.~\eqref{eq:E_log_det_OnePlusAAH_UB} yields the lower bound in Eq.~\eqref{eq:RDF_source_LB}. 

For the upper bound, we consider 
\(\mathds{E}_{\mathbb{P}}[\Delta R_{\mathbb{P},\mathbf{S}}] = \tfrac{1}{2}\mathds{E}\big[\sum_{i=1}^{N} \log\big(\frac{\sigma_{i}}{1+\sigma_{i}}\big)\big]\), 
where \(\{\sigma_{i}\}\) are the eigenvalues of \(\mathbb{A}\mathbb{A}^{\mathsf{H}}\). Using the inequality \(\log\big(\frac{\sigma}{1+\sigma}\big) \leq -\frac{1}{\sigma} + \frac{1}{2\sigma^{2}}\) for \(\sigma>0\), we obtain 
\(\mathds{E}_{\mathbb{P}}[\Delta R_{\mathbb{P},\mathbf{S}}] \leq \tfrac{1}{2}\mathds{E}\big[\sum_{i=1}^{N} \big(-\frac{1}{\sigma_{i}} + \frac{1}{2\sigma_{i}^{2}}\big)\big]\), 
which in matrix form reads 
\begin{align} \label{eq:RDF_source_trace_UB}
\mathds{E}_{\mathbb{P}}[\Delta R_{\mathbb{P},\mathbf{S}}] \leq \tfrac{1}{2}\mathds{E}\big[-\tr\big((\mathbb{A}\mathbb{A}^{\mathsf{H}})^{-1}\big) + \tfrac{1}{2}\tr\big((\mathbb{A}\mathbb{A}^{\mathsf{H}})^{-2}\big)\big].
\end{align}
Substituting \(\mathbb{A} = \mathsf{C}_{\mathbf{H}}^{1/2}\mathbb{P}\), we have
\(\mathds{E}_{\mathbb{P}}[\Delta R_{\mathbb{P},\mathbf{S}}] \leq \tfrac{1}{2}\mathds{E}\big[-\tr(\mathsf{C}_{\mathbf{H}}^{-1}(\mathbb{P}\mathbb{P}^{\mathsf{H}})^{-1})\big] + 
\frac{1}{2}\mathds{E}\big[\tr((\mathbb{P}\mathbb{P}^{\mathsf{H}})^{-1}\mathsf{C}_{\mathbf{H}}^{-1}(\mathbb{P}\mathbb{P}^{\mathsf{H}})^{-1}\mathsf{C}_{\mathbf{H}}^{-1})\big]\). 
To handle the second term, we invoke Von Neumann's trace inequality~\cite{mirsky1975trace}, which yields
\begin{align}
\mathds{E}\big[\tr((\mathbb{P}\mathbb{P}^{\mathsf{H}})^{-1}\mathsf{C}_{\mathbf{H}}^{-1}(\mathbb{P}\mathbb{P}^{\mathsf{H}})^{-1}\mathsf{C}_{\mathbf{H}}^{-1})\big]
\leq \tfrac{\tr ( \mathds{E}[(\mathbb{P}\mathbb{P}^{\mathsf{H}})^{-2}])}{\lambda_{\min}^{2}(\mathsf{C}_{\mathbf{H}})}.
\end{align}
It follows from~\cite[Lemma 2.10]{tulino2004random} and~\cite[21.14]{seber2008matrix} that
\begin{align}
   \mathds{E}\big[(\mathbb{P}\mathbb{P}^{\mathsf{H}})^{-1}\big] &= \frac{\mathsf{I}_{N}}{p (n_{t} - m_{t})}, \label{eq:Wishart_inverse1} \\
   \mathds{E}\big[(\mathbb{P}\mathbb{P}^{\mathsf{H}})^{-2}\big] &= \frac{n_{t}\,\mathsf{I}_{N}}{p^{2} (n_{t} - m_{t}-1)(n_{t} - m_{t})(n_{t} - m_{t}+1)}, \label{eq:Wishart_inverse2}
\end{align}
for \(n_{t} \geq m_{t}\). By applying Eq.~\eqref{eq:RDF_source_trace_UB} and substituting \(p = \frac{\mathrm{SNR}_{\mathrm{dl}}}{m_{t}}\), the corresponding upper bound can be established in Eq.~\eqref{eq:RDF_source_UB}.

Next, we analyze the asymptotic behavior of the bounds with respect to \(n_{t}\).  
Starting with the lower bound, and using the asymptotic expansion~\cite[6.3.18]{abramowitz1964handbook}, we have
\begin{align}
\sum\nolimits_{j=0}^{m_{t}-1} \psi(n_{t} - j)
= m_{t} \log n_{t} - \frac{m_{t}^{2}}{2n_{t}} + O\!\left(\frac{1}{n_{t}^{2}}\right).
\end{align}
Moreover, by expanding the logarithmic term and expressing it in terms of the trace of \(\mathsf{C}_{\mathbf{H}}\), we obtain
\begin{align}
& \log \det\big(\mathsf{I}+\mathrm{SNR}_{\mathrm{dl}} \tfrac{n_{t}}{m_{t}}\,\mathsf{C}_{\mathbf{H}}\big) = N\log p + \log \det(\mathsf{C}_{\mathbf{H}}) \notag \\
& \quad \quad \quad \quad \quad \quad \quad + N \log n_{t} + \frac{\tr(\mathsf{C}_{\mathbf{H}}^{-1})}{p n_{t}} + O\!\left(\frac{1}{n_{t}^{2}}\right).
\end{align}
This yields \(C_{\mathbf{S}}^{L}\) in Eq.~\eqref{eq:RDF_source_LB_coefficient}.
On the other hand, \(C_{\mathbf{S}}^{U}\) in Eq.~\eqref{eq:RDF_source_UB_coefficient} follows directly from the first term of Eq.~\eqref{eq:RDF_source_UB}.

\section{
Proof of Theorem~\ref{thm:RDF_increase} \label{proof:thm:RDF_increase}
}
We have the quantity of interest \(D_{\mathrm{mmse}|\mathbb{P}} = \tr\big(\mathsf{C}_{\mathbf{H}}(\mathsf{I} + \mathbb{A}\mathbb{A}^{\mathsf{H}})^{-1}\big)\), and the function \(f(\mathsf{X}) = \tr\big(\mathsf{C}_{\mathbf{H}}(\mathsf{I}+\mathsf{X})^{-1}\big)\) is convex for \(\mathsf{X} \succeq \mathbf{0}\). By applying Jensen's inequality and recalling the argument from Eq.~\eqref{eq:E_log_det_OnePlusAAH} to Eq.~\eqref{eq:E_log_det_OnePlusAAH_UB}, we obtain the lower bound in Eq.~\eqref{eq:Dmmse_LB}.
For the upper bound, we first note that \(\tr\big(\mathsf{C}_{\mathbf{H}}(\mathsf{I} + \mathbb{A}\mathbb{A}^{\mathsf{H}})^{-1}\big) \leq \tr\big(\mathsf{C}_{\mathbf{H}}(\mathbb{A}\mathbb{A}^{\mathsf{H}})^{-1}\big)\). Recalling Eq.~\eqref{eq:Wishart_inverse1}, we then have the upper bound in Eq.~\eqref{eq:Dmmse_UB}.
For the asymptotic analysis, we consider the lower bound 
\(\tr\big(\mathsf{C}_{\mathbf{H}} (\mathsf{I} + \alpha \mathsf{C}_{\mathbf{H}})^{-1}\big) = \sum_{i=1}^{N} \frac{\lambda_{i}}{1+\alpha \lambda_{i}}\), 
where \(\alpha = \mathrm{SNR}_{\mathrm{dl}} \frac{n_{t}}{m_{t}}\) and \(\{\lambda_{i}\}\) are the eigenvalues of \(\mathsf{C}_{\mathbf{H}}\). For large \(n_{t}\), we have 
\(\frac{\lambda_{i}}{1+\alpha \lambda_{i}} = \frac{1}{\alpha} + O(\frac{1}{n_{t}^{2}})\). Combining this with the upper bound, we obtain the asymptotic scaling in Eq.~\eqref{eq:Dmmse_asymptotics}.

Recall that when \(n_{c} = n_{p}\) and \(n_{t} \geq m_{t}\), we have 
\(\mathds{E}_{\mathbb{P}}[\Delta R_{\mathbb{P},d}] = \frac{N}{2} \, \mathds{E}_{\mathbb{P}}\big[-\log\big(1 - \tfrac{D_{\mathrm{mmse}|\mathbb{P}}}{d}\big)\big]\). 
Applying Jensen's inequality and using the lower bound of \(\mathds{E}_{\mathbb{P}}[D_{\mathrm{mmse}|\mathbb{P}}]\), a lower bound on \(\mathds{E}_{\mathbb{P}}[\Delta R_{\mathbb{P},d}]\) follows in Eq.~\eqref{eq:RDF_distortion_LB}.
For the upper bound, we first introduce the notation \(\mathbb{X} = \mathsf{C}_{\mathbf{H}}(\mathsf{I}+\mathbb{A}\mathbb{A}^{\mathsf{H}})^{-1}\). 
Applying \(-\log(1-a) \leq \frac{a}{1-a}\) for \(a \in (0,1)\), the quantity of interest can be upper-bounded as
\begin{align}
-\tfrac{N}{2}\,\mathds{E}\big[\log\big(1-\tfrac{1}{d}\tr(\mathbb{X})\big)\big]
&\leq \tfrac{N}{2}\,\mathds{E}\big[\big(\tfrac{d}{\tr(\mathbb{X})}-1\big)^{-1}\big].
\end{align}
Recall that the distortion constraint satisfies \(\tr(\mathbb{X}) \leq d \leq \tr(\mathbb{X}) + N \lambda_{\min}(\mathsf{C}_{\mathbf{S}|\mathbb{P}})\) in Eq.~\eqref{eq:distortion_range}.  
To simplify the bound, we impose the condition
\begin{align}
d \geq \tfrac{1}{1-\epsilon}\,\tr(\mathbb{X}),
\end{align}
for some \(\epsilon \in (0,1)\).  
We can now proceed to further bound
\begin{align}
\mathds{E}\big[\big(\tfrac{d}{\tr(\mathbb{X})}-1\big)^{-1}\big]
&\leq \tfrac{1}{\epsilon d}\mathds{E}\big[\tr(\mathbb{X})\big].
\end{align}
Applying the upper bound on \(\mathds{E}[\tr(\mathbb{X})]\) from Eq.~\eqref{eq:Dmmse_UB} yields an upper bound on \(\mathds{E}_{\mathbb{P}}[\Delta R_{\mathbb{P},d}]\) as given in Eq.~\eqref{eq:RDF_distortion_UB}.
In the regime of large \(n_{t}\), the asymptotic coefficient can be readily obtained from the previously established asymptotic analysis of \(D_{\mathrm{mmse}}^{L}\) and \(D_{\mathrm{mmse}}^{U}\).
\bibliographystyle{IEEEtran}
\bibliography{references}

\end{document}